\documentclass[11pt]{article}

\usepackage{amsfonts,amssymb,amsmath,latexsym,xfrac,ae,aecompl}

\usepackage{graphicx}
\usepackage{xcolor}

\newcommand{\BBB}{\vspace*{-\bigskipamount}}

\newcommand{\cO}{\mathcal{O}}

\newcommand{\Paragraph}[1]{\BBB\paragraph{#1}}
\newcommand{\remove}[1]{}

\newlength{\pagewidthA}
\newlength{\pagewidthC}
\newlength{\captionwidth}
\newcommand{\qed}{\hfill $\square$ \smallbreak}
\newenvironment{proof}{\noindent{\bf Proof:}}{\qed}

\newtheorem{theorem}{Theorem}

\begin{document}

\baselineskip           	3ex
\parskip                	1ex

\title{Routing in Wireless Networks with Interferences \footnotemark[1] \vfill}

\author{
Bogdan S. Chlebus \footnotemark[2]
\and
Vicent  Cholvi \footnotemark[3]
\and
Pawel Garncarek \footnotemark[4]
\and
Tomasz Jurdzi\'{n}ski \footnotemark[4]
\and
Dariusz R. Kowalski \footnotemark[5]}

\footnotetext[1]{This paper was published as~\cite{ChlebusCGJK-IEEE2017}. 
The work was supported by the Polish National Science Centre under grant  DEC-2012/06/M/ST6/00459 and by the Spanish Ministry of Education Culture and Sport under grant PRX16/00086.}

\footnotetext[2]{Department of Computer Science and Engineering, University of Colorado Denver, Denver, Colorado, USA.}

\footnotetext[3]{Departament de Llenguatges i Sistemes Inform\`{a}tics, Universitat Jaume I, Castell\'{o}, Spain.}

\footnotetext[4]{Instytut Informatyki, Uniwersytet Wroc{\l}awski, Wroc{\l}aw, Poland.  }

\footnotetext[5]{Department of Computer Science, University of Liverpool, Liverpool, United Kingdom.}

\date{}

\maketitle

\vfill

\begin{abstract}
We consider dynamic routing in multi-hop wireless networks with adversarial traffic.
The model of wireless communication incorporates interferences caused by packets' arrivals  into the same node that overlap in time. 
We consider two classes of adversaries: balanced and unbalanced.
We demonstrate that, for each routing algorithm and an  unbalanced adversary, the algorithm is unstable against this adversary in some networks.
We develop a routing algorithm that has bounded packet latency against each balanced adversary.

\vfill

~

\end{abstract}

\vfill

\thispagestyle{empty}

\setcounter{page}{0}

\newpage

% 	end titlepage

\section{Introduction}

Models of wireless data networks that abstract from incidental systems details and concentrate on the essential aspects of communication are most conducive to studying routing algorithms. 
One of such aspects are interferences.
The model of \emph{radio networks}~\cite{ChlebusKPR11} assumes that when multiple packets arrive simultaneously into a node then this results in interference experienced by the receiving node.
Such networks are considered in this paper.

Adversarial methodologies of traffic generation make it possible to consider worst-case behavior of routing. 
We use such an approach to study routing in radio networks.
Such networks pose unique challenges to design of routing algorithms because of the need to coordinate activities of the nodes whose transmissions may reach  some node simultaneously.

\Paragraph{Related work.}
The methodology of adversarial routing  in wired networks was pioneered by Borodin et al.~\cite{BorodinKRSW01} and Andrews et al.~\cite{AndrewsAFLLK01}.
Lotker et al.~\cite{LotkerPR04} showed that, in wired networks, every greedy scheduling policy is stable  if the injection rate is smaller than $1/(L + 1)$, where $L$ is the length of the longest route used by any packet.

Stability in general wireless networks without explicit interferences was studied by Andrews and Zhang \cite{AndrewsZ-JACM05,AndrewsZ07} and Cholvi and Kowalski~\cite{CholviK10}. 
Lim et al.~\cite{LimJA14} analyzed the stability of the max-weight protocol in wireless networks  with interferences, but assuming the existence of a set of feasible edge rate vectors sufficient to keep the network stable. 

Chlebus et al.~\cite{ChlebusKR-TALG12} and Anantharamu et al.~\cite{AnantharamuCKR-INFOCOM10} studied adversarial broadcasting in the case of using single-hop radio networks. 
Chlebus et al.~\cite{ChlebusCK-FOMC14} considered interactions among components of routing in wireless networks, which included transmission policies, scheduling policies to select the packet to transmit from a set of packets parked at a node, and hearing control mechanisms to coordinate transmissions with scheduling. 

\Paragraph{Our results.}
We study dynamic routing in multi-hop radio networks with a specific methodology of adversarial traffic that reflects interferences. 
We demonstrate that there is no routing algorithm guaranteeing stability for an injection rate greater than~$1/L$, where the adversary's parameter~$L$ is the largest number of links which a packet needs to traverse while routed to its destination. 
We give a routing algorithm that guarantees stability for injection rates smaller than $1/L$.

\section{Routing against Interferences}
\label{sec:model}

We consider communication in multi-hop radio networks. 
A network is modeled as a (simple) undirected connected graph $G=(V,E)$ with some $n=|V|$ nodes.
An edge in $E$ represents two directed communication channels connecting the endpoints; an oriented edge from $E$ is referred to as a \emph{link}.
An edge~$(u,v)$, when interpreted as a link with tail~$u$ and head~$v$, is denoted as $u\rightarrow v$.
Messages are transmitted along the links according to the links' orientation.
At most one link determined by an edge can be used at a time.

We say that some parameter of the communication environment is \emph{known} when it can be used in an algorithm's code.
Each node is assigned a unique name, which is an integer in $[1; n]$. 
Every node knows $n$ and its own name.

An execution of a communication algorithm is synchronous, in that it is structured as a sequence of \emph{rounds}. 
In each round, a node may either \emph{transmit} a message or \emph{listen} trying to hear incoming messages. 
Messages are delivered in the round of transmission.
A message that is successfully received is said to be \emph{heard} by the receiving node.
A node $v$ can hear a message from its neighbor~$u$ in a round $t$ if $v$ listens in round~$t$ and $u$ is the only node among $v$'s neighbors  that transmits in this round. 
Messages delivered simultaneously to a node but not heard by the node are said to \emph{collide} or \emph{interfere with one another} at the node.
Messages facilitate routing, in particular they may carry packets traversing the network.
Nodes may need to store multiple packets in their private memory, which is referred to as the node's \emph{queue}.
The number of packets residing simultaneously in such a queue is the queue's \emph{size}.

\subsection{Routing}

A routing algorithm handles packets that are injected at the nodes of graph $G$ and need to reach their respective  destination nodes by traveling through the network in a store-and-forward manner.
A packet together with the round it was injected in and the (simple) oriented path it needs to traverse, make a \emph{tour}.
Each packet is encapsulated as a part of its tour at the time of injection and during  the network's traversal.
The number of links in a tour's path is this path's \emph{length},  also referred to as the \emph{tour's length}.

Consider a tour determined by a packet $p$ injected in round~$t$ into node~$v_1$ that needs to pass through the nodes on the path  $\langle v_1,\ldots,v_k\rangle $ to reach $v_k$.
Each link $v_i\rightarrow v_{i+1}$ is said to be among the tour's \emph{links}, for $1\le i<k$.
The start node~$v_1$ is this tour's \emph{source} and the end node~$v_k$ is this tour's \emph{destination}.
Packet~$p$ is routed by traversing the links according to the tour's specification: there are $k-1$ rounds $t_1,\ldots,t_{k-1}$ such that the node~$v_i$ transmits packet~$p$ in round~$t_i$ and the node~$v_{i+1}$ hears~$p$ in this very~round~$t_i$,  for $i\in[1,k-1]$, where $t\leq t_1<t_2<\cdots<t_{k-1}$.
The packet is \emph{delivered} in round~$t_{k-1}$, which is the round when the packet reaches the tour's destination.
The \emph{latency} of this specific routing of the tour is $t_{k-1}-t$, which is the number of rounds the tour spends in the network between its packet's injection at the source and its delivery to the destination.

\subsection{Interferences}

We represent interferences as abstract conflicts between parts of a network.
The basic case is of a conflict between a node $w$ with a link $u\rightarrow v$.
Intuitively, it occurs when a transmission by~$w$  cannot be reconciled with having a different message delivered successfully from~$u$ to~$v$ in the same round.
A node~$w$ \emph{conflicts with a link $u\rightarrow v$} if either  $w=u$ or $w=v$ or the nodes~$v$~and~$w$ are neighbors.
We extend this to say that a node~$w$ \emph{conflicts with a tour} if $w$ conflicts with some among the tour's links. 
This concept of conflict between a node and a tour conservatively reflects the worst-case possibility of the node's transmission interfering with the tour's packet when it is traversing the tour's path.

When transmissions through links of a tour may prevent packet deliveries along the links of another tour then these tours are said to be in conflict.
Formally, tours $f_0$ and~$f_1$ \emph{conflict with one another} if either they pass through the same node or there is a node in one of these tours $f_i$, which is different from~$f_i$'s destination,  that conflicts with the other tour~$f_{1-i}$.
Conflicting pairs of tours can be interpreted as edges in a new graph.
Formally, for a simple graph $G=(V,E)$ and a set of tours~$F$ in~$G$, the \emph{conflict graph of~$F$} is a simple graph with tours in $F$ taken as the vertices and two different tours from $F$ making an edge when they conflict with each other.

Figure~\ref{fig:conflicting_graph} gives an example of a conflict graph  for a network and a set of tours in it, where arrows represent the traversed paths.
The graph on the left represents a network in which routing is performed, and the resulting conflict graph for the specified tours is depicted on the right.

%: figure two graphs

\begin{figure}
\begin{center}
\hfill
\includegraphics[width=\pagewidthA]{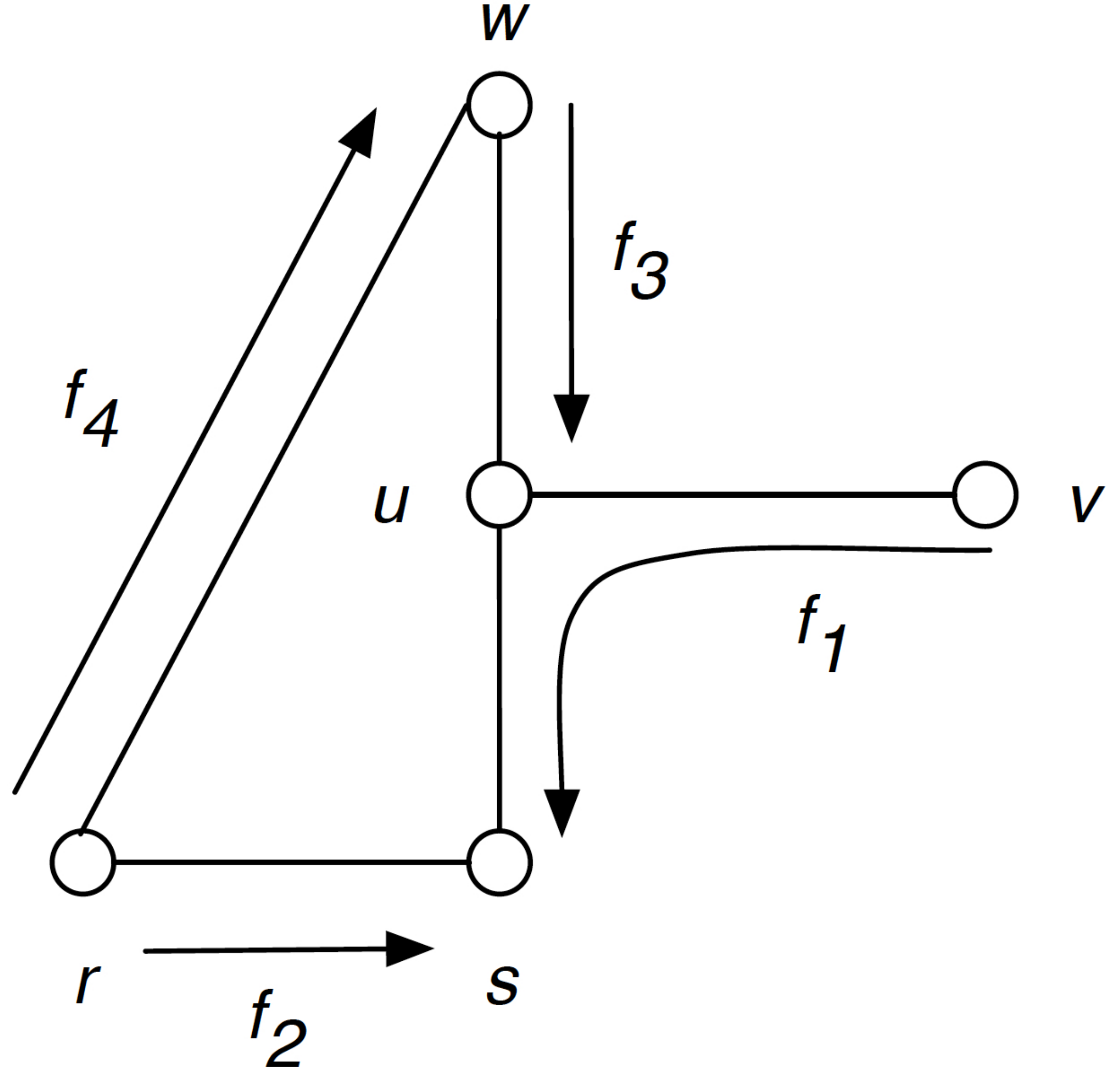}
\hfill
\includegraphics[width=\pagewidthA]{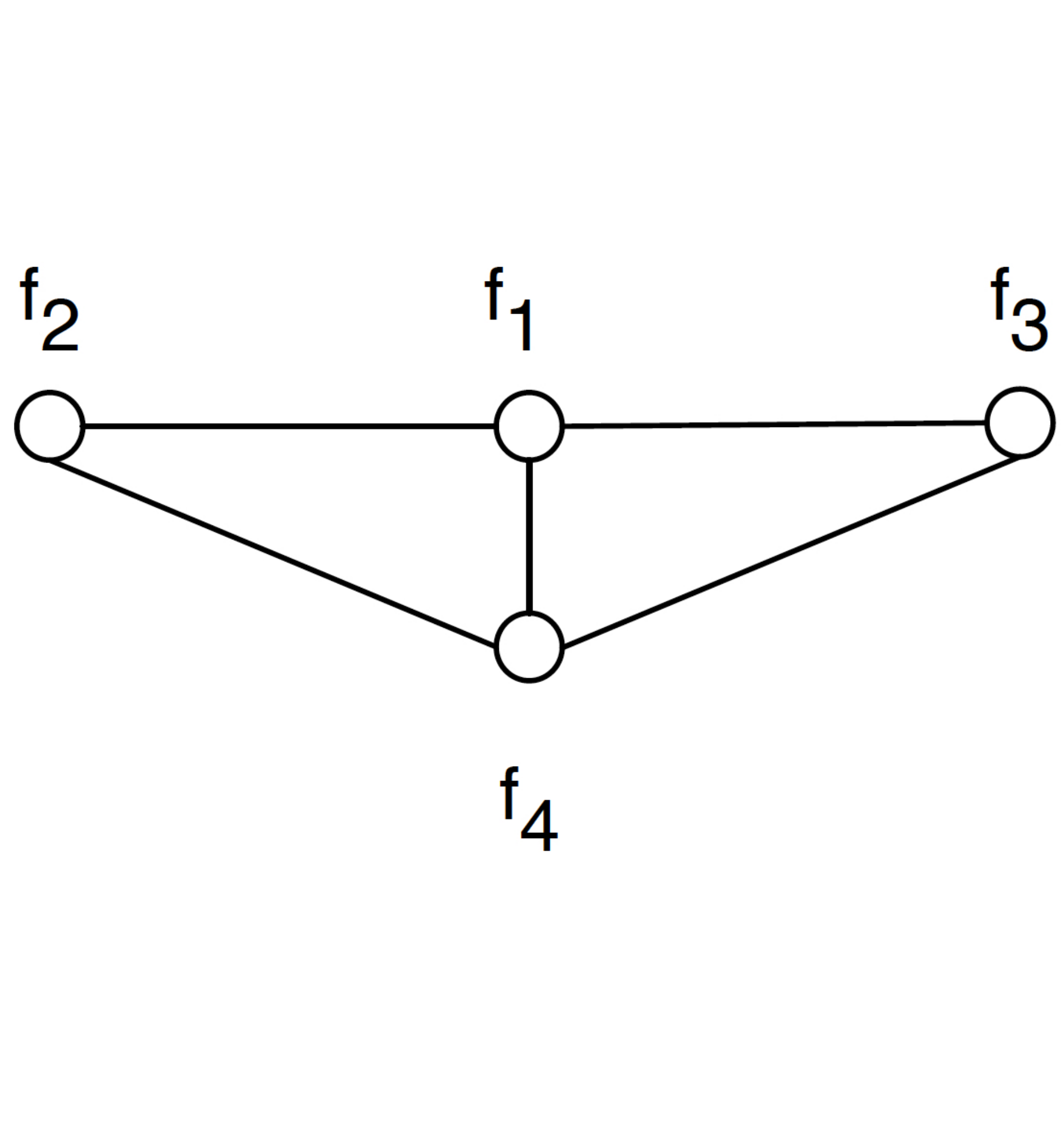}
\hfill ~

\parbox{\captionwidth}{\caption{\label{fig:conflicting_graph}
Four tours named $f_1, f_2, f_3, f_4$.
Tour~$f_1$ conflicts with each of the other tours~$f_i$, for $i\ne 1$, because it shares nodes with both $f_2$ and $f_3$, and the start node $r$ of $f_4$ conflicts with the link $u\rightarrow s$, which belongs to~$f_1$.
 Tour~$f_2$ conflicts with $f_4$ because they share the node~$r$, while $f_2$ does not conflict with $f_3$ because they do not share nodes and only their respective destination endpoints $s$ and $u$ conflict with the other tour.
 Tours~$f_3$ and~$f_4$ conflict with each other because they share the node~$w$.}}
 \end{center}
\end{figure}

\subsection{The static link scheduling problem}

We now consider static routing, when a set of routing tours  is given as input.
(This is an aside to the main topic of dynamic routing.)
The problem is further restricted such that each tour is just one link.
The goal is to route all these tours within the shortest possible interval of rounds.
This is known as the \emph{static link scheduling (SLS)} problem~\cite{Kesselheim12}. 
We show that SLS  is related to  vertex-colorings of conflict graphs.
The minimum number of colors used in coloring vertices of a graph, such that each pair of adjacent vertices are assigned distinct colors,  is the \emph{chromatic number} of the graph.

\begin{theorem}
\label{the:optimal_transmission}
The minimum number of rounds to route an instance of  the link scheduling problem is equal to the chromatic number of the respective conflict graph.
\end{theorem}

\begin{proof}
Consider a set $F$ of one-link  tours that makes an instance of the static link scheduling problem. 
Let $\mu$ be the chromatic number of the respective collision graph.
Let $T$ be the number of rounds of a shortest schedule $S$ to route the packets in $F$.
We want to show that $T=\mu$.

First we show that $\mu\le T$. 
We may assume without loss of generality that each packet is transmitted exactly once in $S$. Namely, we may replace $S$ by a minimal subset of transmissions in $S$ such that each packet's  transmission results in this packet's delivery to its destination.
This implies that packets of conflicting tours are transmitted in different rounds. 
Now assign the round of transmission of a packet as the color of the node representing this packet's tour. 
This is a proper coloring because an edge connects two conflicting tours, which  are transmitted in different rounds.
Therefore there exists a vertex coloring of the conflict graph of $T$ colors, and so $\mu\le T$.

Next we show that $T\le \mu$.
Consider a coloring of the conflict graph with $\mu$ colors.
The colors could be identified with the integers in the interval $[1,\mu]$.
A schedule of transmissions in $\mu$ rounds can be defined as follows: a node of color~$i$ transmits in round~$i$.
When a node transmits a packet then the packet is heard by the destination node because no packets with conflicting tours are transmitted in this round, as their tours are of different colors. 
It follows that all the packets get delivered in $\mu$ rounds, so that $T\le \mu$.
\end{proof}

\subsection{Adversaries inject  tours}

We model dynamic injection of tours by way of an adversarial model, in the spirit of similar approaches used in~\cite{BorodinKRSW01,AndrewsAFLLK01,LotkerPR04,ChlebusKR-TALG12,AnantharamuCKR-INFOCOM10,CholviK10,ChlebusCK-FOMC14}.
An adversary represents the users that generate packets to be routed in a given radio network.
The constraints imposed on packet generation by the adversary allow to consider worst-case performance of deterministic routing algorithms handling dynamic traffic.

An adversary is determined by three numbers $(\rho,b,L)$, which together are called its \emph{type}.
The number $\rho$ is the \emph{injection rate} and needs to satisfy $0 \leq \rho \leq 1$.  
The number~$b>0$ is called the \emph{burstiness} and  represents the maximum number of tours injected in the same round that may conflict with a node. 
The number~$L$ is the \emph{stretch} and is a positive integer.

An adversary injects packets continuously by placing them in the nodes of a network.
An injected packet is encapsulated with the path it needs to traverse  into a tour.

Let $\tau$ be a time interval and $v$ a node.
We refer to the number of tours injected during~$\tau$ that $v$ conflicts with as the \emph{load of node~$v$ in~$\tau$}.
This means that a tour $f$ contributes a unit to the load of each node $v$ such that $v$ conflicts with~$f$.  

The adversary of a given injection rate $\rho$, burstiness~$b$ and stretch~$L$ is subject to the following restrictions in how tours may be injected.
First, for each time interval $\tau$ and each node~$v$ the load of~$v$ in~$\tau$ is at most
$\rho \,\cdot \!\mid\!\tau \!\mid \! + \, b$.
Second, the length of path of an injected tour is at most~$L$. 

An adversary of type $(\rho,b,L)$ is called \emph{balanced} when the inequality $\rho \cdot L < 1$ holds, and it is \emph{unbalanced} when the inequality $\rho \cdot L > 1$ holds.

\section{Stability}

A routing algorithm handles itineraries, which include the paths the packets are to traverse.
A critical part of a routing algorithms is a \emph{transmission policy} which determines which nodes transmit in a round,  as well as the contents of the transmitted messages.
We consider distributed transmission policies, when each node decides in each round whether to transmit a message, and if so, then it determines the contents of the transmitted message.
A message contains one tour stored in the node's queue, and possibly additional control bits.

The immediate goal of transmitting the packet of a tour is to forward it to the node designated in the tour as the next one on the path to be traversed by the tour's packet.
A transmitted tour is heard by the intended recipient node when that recipient node is not transmitting in this round and the transmitted message does not interfere with other transmissions, following the radio network model's specification, as given in Section~\ref{sec:model}.

A routing policy is \emph{stable} against an adversary when the number of tours in queues at the nodes is bounded in all executions when packet injections conform to the restrictions imposed by the adversary's type.

 \emph{Packet latency} of a routing algorithm against an adversary is  the maximum latency attained by a tour in executions subject to the restrictions imposed by the adversary's type.
When packet latency is bounded then queues are also bounded, as each queue size in a node is the lower bound on the delay of a packet already queued.

\begin{theorem}
For each  unbalanced adversary and each sufficiently large integer $n>0$ there exists a network of $n$ nodes such that every routing algorithm is unstable when the adversary injects tours into this network.
\end{theorem}

\begin{proof}
Let $(\rho,b,L)$ be the type of an unbalanced adversary, which means such that $\rho \cdot L > 1$.
Take an arbitrary $n>L$ and let the network be the clique of $n$ nodes.

An injected tour of a positive length contributes a unit to each node's load, because all nodes are neighbors. 
Therefore, the adversary can inject up to $\rho t +  b$ tours into all the nodes in the network in a time interval of length~$t$.
When there are multiple disjoint time intervals of length $t$ each, then the burstiness~$b$ can be accounted for at most once, but up to $\rho t$ new packets can be injected in each such an interval.

The adversary will inject packets that need to be forwarded exactly~$L$ times each. 
In a complete network, at most one message can be heard in a round, so at most one packet can be forwarded in a round.
As each tour contains exactly $L$ links, the total number of message to be heard, in order to deliver the packets  injected in a time interval of length~$t$, is  at least
\[
L  \rho t   = t + (L\rho -1) t.
\ 
\] 
At most $t$ messages can be heard in $t$ rounds, so  the adversary can generate,  in disjoint intervals of $t$ rounds, a surplus of $(L\rho -1) t$ messages to be heard in the future.
Take an integer~$t$ such that $(L\rho -1) t\ge 1$, which exists because $\rho \cdot L > 1$.

At least one packet is needed to account for $L$ messages. 
Thus the adversary can make the number of packets in the queues grow by at least one packet per $L$ intervals in a sequence of consecutive disjoint time intervals of $t$ rounds each.
The resulting execution is unstable.
\end{proof}

\section{Efficient Routing}

We specify a routing algorithm that we call \textsc{Old-Go-First}.
It provides bounded packet latency when executed against balanced adversaries.
In the algorithm's design, we rely on the Brook's theorem~\cite{Brooks1941}, which states that a graph of a maximum node degree $\Delta$ can be colored with $\Delta+1$ colors. 

The algorithm's execution is partitioned into disjoint intervals of rounds called \emph{windows}.
Tours injected in a window  are \emph{new} in this window and become \emph{old} when the next window starts.
Tours that are old in a window are transmitted in the window, while the new tours wait for the next window to be transmitted in it as old.

A window is partitioned into two phases, see Figure~\ref{fig:transmission-protocol}.
The first phase consists of preprocessing in order to prepare the second phase, which is spent  executing a transmission policy.

The phases are specified in greater detail next.

\noindent
\emph{Phase~1}: preprocessing to prepare routing in the next phase:  

\begin{enumerate}
\item
Collect in each node the specification of all the old tours. 
Let $L'$ denote the length of the longest old tour.
\item
Build the conflict graph for the old tours.
Let $\Delta$ be the maximum degree of a node in the conflict graph.
\item
Color the vertices of this conflict graph with $\Delta+1$ colors. 
\end{enumerate}

\noindent
\emph{Phase~2}: tours are routed in the next time interval of $L' (\Delta+1)$ rounds by the following transmission policy:

\begin{enumerate}
\item
Partition this time interval into $L'$ intervals, called \emph{super-rounds}, each of ${\Delta + 1}$ rounds.
\item
In each super-round: a  node storing  a tour of color $i$ transmits this tour in the 
$i$th round of the super-round.
\end{enumerate}

\begin{figure}
\begin{center}
\includegraphics[width=\pagewidthC]{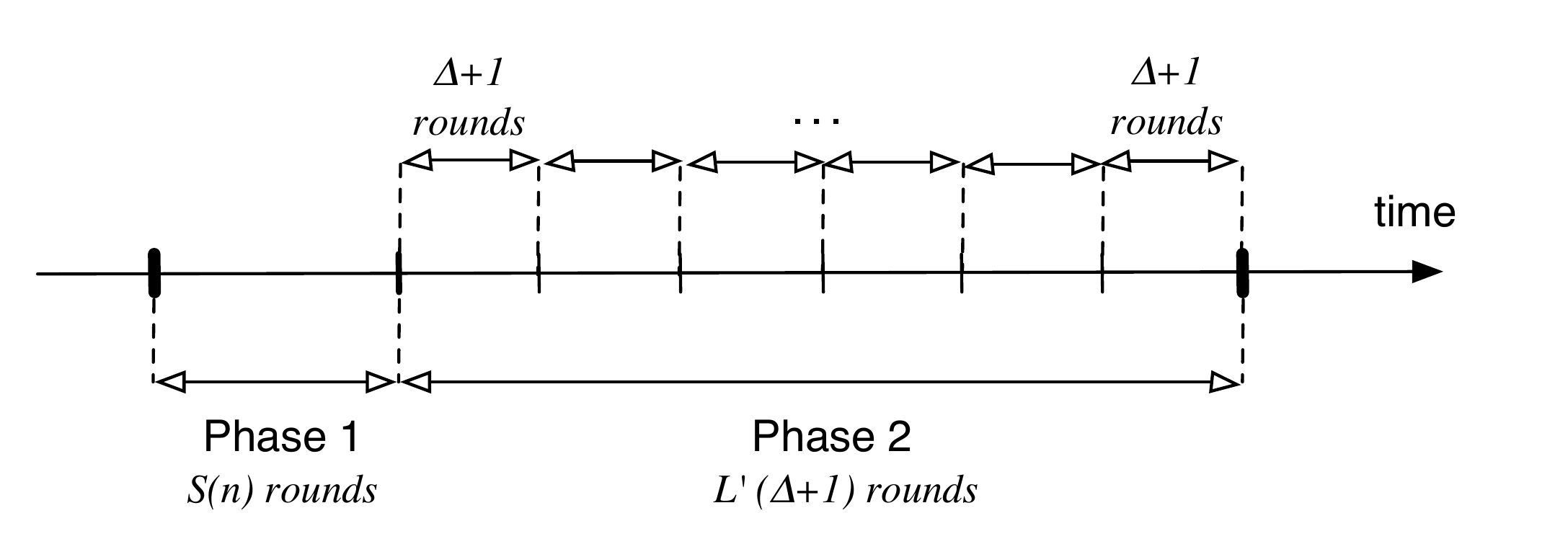}

\parbox{\captionwidth}{\caption{\label{fig:transmission-protocol}
A  window of an execution of transmission policy \textsc{Old-Go-First} consists of two phases.
Phase one is preprocessing.
Phase two consists of $L'$ super-rounds, each of $\Delta+1$ rounds, where $\Delta$ is the maximum degree of the conflict graph of the old tours in this window, which was built in phase one.}}
\end{center}
\end{figure}

A node stores at most one tour of each color in the beginning of a super-round, because tours passing through a node conflict with each other, so they are colored differently.
Collecting the information about all the old tours originating in each node can be accomplished by gossiping of what each node stores originally.
We can use the algorithm given in~\cite{ChlebusKPR11}, which runs in $S(n)=\cO( n\log^4 n)$ rounds in a network of~$n$ nodes.
This algorithm uses ``short'' messages, in that one rumor (the information to be gossiped that originates in one node) requires one message, and it takes one round to transmit a message.

For the purpose to prepare a transmission policy for a window, a rumor contains the information about all the old tours stored in a node at the end of the previous window.
Once gossiping is completed, building the conflict graph and coloring its nodes~\cite{SKULRATTANAKULCHAI2006} can be done in negligible time,  concurrently by all the nodes.

\begin{theorem}
Routing algorithm \textsc{Old-Go-First} attains packet latency $\cO\bigl( \frac{n\log^4 n +bL}{1-\rho L} \bigr)$ when executed against a balanced adversary of type $(\rho,b,L)$ on a network of $n$ nodes.
\end{theorem}

\begin{proof}
The number of  tours injected in a window of length $w$ that contribute to the load of a node is at most $\rho w +  b$, each of stretch at most~$L$. 
By the design of \textsc{Old-Go-First},  in a super-round, each tour, that is still on its way, becomes colored and so is transmitted at some round of this super-round.
The message with such a tour is heard immediately by the receiving node, by the definition of coloring.
Since  phase two consists of $L'$ super-rounds, each old tour is delivered to its destination within this window.
Window size $w$ needs to be sufficiently large to accommodate the two phases.
By the design of the two phases, it is sufficient if $w$ satisfies the following inequality:
\begin{equation}
\label{eqn:super-round}
 S(n) + (\rho w + b) \cdot L \le w
\ .
\end{equation}
This is because phase one takes $S(n)$ rounds, and $\Delta + 1\le \rho w + b$, as at most these many tours conflicting with a node can be injected during the previous window.
Bound~\eqref{eqn:super-round} is equivalent to $S(n) + b L \le w(1-\rho L)$,  by algebra.
Since the inequality $\rho L <1$ holds, by the assumption that the adversary is balanced, we may take the following quantity
\[
u=\Bigl\lceil \frac{S(n)+ b  L}{1-\rho L}\Bigr\rceil
\]
as an upper bound of every window~$w$.
Packet latency is at most $2u$, because a packet injected in the beginning of a window is delivered by the end of the next window.  
Using the estimate $S(n)=\cO\left( n\log^4 n \right)$, we conclude that $\cO\bigl( \frac{n\log^4 n +bL}{1-\rho L} \bigr)$ is an upper bound on packet latency.
\end{proof}

\section{Conclusion}

We proposed an adversarial framework to study stability of deterministic distributed routing algorithms in multi-hop wireless networks with interferences. 
It is representative enough  to deny stability for sufficiently strong adversaries, namely, the unbalanced ones.
We showed that there exists a deterministic distributed routing algorithm that provides bounded packet latency against balanced adversaries in all connected radio networks.
This algorithm needs to know the size of the network, which is required in gossiping, but  the adversary does not need to be known.
Theorem~\ref{the:optimal_transmission} implies that achieving optimal packet latency for static instances of routing, in the case when packets need to make one hop only, is equivalent to finding the chromatic number of the respective conflict graph.  
One can show that conflict graphs have sufficiently expressive topologies to make the problem of their coloring NP-hard, as is the case for all simple graphs; we omit the details.

%: bibliography

\bibliographystyle{abbrv}

\bibliography{radio-route}

\end{document}